\documentclass[letterpaper, 10 pt, conference]{ieeeconf}
\IEEEoverridecommandlockouts
\usepackage{cite}
\usepackage[compact]{titlesec}
\usepackage{amsmath,amssymb,amsfonts}
\usepackage{hyperref}
\usepackage{graphicx}
\usepackage{textcomp}
\usepackage{times}
\usepackage{graphicx}
\usepackage{mathrsfs}
\usepackage{amstext} 
\usepackage{algorithm}
\usepackage{algorithmic}
\usepackage{setspace} 
\usepackage{amssymb} 
\usepackage{mathtools} 
\usepackage{verbatim}
\usepackage{color}
\usepackage{booktabs}
\usepackage[font=small]{caption}

\renewcommand{\theenumi}{\alph{enumi}}

\newcommand{\naturals}{\mathbb{N}}
\newcommand{\real}{\mathbb{R}}

\newcommand{\Bc}{\mathcal{B}}

\newcommand{\Ec}{\mathcal{E}}
\newcommand{\Fc}{\mathcal{F}}
\newcommand{\Gc}{\mathcal{G}}
\newcommand{\Hc}{\mathcal{H}}
\newcommand{\Ic}{\mathcal{I}}
\newcommand{\Jc}{\mathcal{J}}
\newcommand{\Kc}{\mathcal{K}}

\newcommand{\Pc}{\mathcal{P}}

\newcommand{\Sc}{\mathcal{S}}

\newcommand{\Xc}{\mathcal{X}}

\newcommand{\longthmtitle}[1]{\mbox{}{\textit{(#1):}}}

\newcommand{\oprocendsymbol}{\hbox{$\square$}}
\newcommand{\oprocend}{\relax\ifmmode\else\unskip\hfill\fi\oprocendsymbol}

\newtheorem{theorem}{Theorem}[section]
\newtheorem{definition}[theorem]{Definition}

\newtheorem{remark}[theorem]{Remark}

\newtheorem{proposition}[theorem]{Proposition}

\newcommand{\marginJC}[1]{\marginpar{\color{red}\tiny\ttfamily JC: #1}}
\newcommand{\marginDS}[1]{\marginpar{\color{blue}\tiny\ttfamily DS: #1}}

\newcommand{\R}{\mathbb{R}}

\renewcommand{\baselinestretch}{0.97}
\graphicspath{{figures/}}

\usepackage{multirow}  
\usepackage{subcaption}  

\allowdisplaybreaks

\def\BibTeX{{\rm B\kern-.05em{\sc i\kern-.025em b}\kern-.08em
    T\kern-.1667em\lower.7ex\hbox{E}\kern-.125emX}}
\begin{document}
\title{\LARGE Learning the Omega-Limit Partition from Approximate
  Koopman Invariants
  \thanks{This work was supported by AFOSR Award FA9550-23-1-0740.
}}

\author{Dhruv Shah \quad Jorge Cort\'es \thanks{D. Shah and
    J. Cort\'es are with Department of Mechanical and Aerospace
    Engineering, UC San Diego, USA, {\tt\small
      \{dhshah,cortes\}@ucsd.edu}}}

\maketitle
\thispagestyle{empty}

\begin{abstract}
  Nonlinear dynamical systems with multiple omega-limit sets pose a
  fundamental challenge for global finite-dimensional Koopman modeling
  with continuous observables. A natural way to address this challenge
  is to decompose the state space into forward-invariant regions
  associated with distinct asymptotic behaviors and to learn a
  separate Koopman model for each region. This paper studies how to
  compute such partitions from data. We establish conditions under
  which Koopman-invariant observables, namely eigenfunctions
  associated with the unit eigenvalue, characterize the stationary
  partition induced by omega-limit sets. {Building on this
  characterization, we develop OPAL (Omega-limit Partitions from
  Approximate invariants Learned from data), an algorithm that
  extracts partitions from approximate invariant observables.} Numerical
  examples show that the recovered partitions capture the relevant
  asymptotic regions and support multiple regional linear Koopman
  models with significantly improved long-horizon prediction capabilities.
\end{abstract}

%

%

\section{Introduction}
The Koopman operator provides a linear representation of nonlinear
dynamics by lifting the evolution from the state space to a space of
observable functions. This perspective has become increasingly
influential in data-driven approaches to dynamical systems because it
offers a principled way to combine nonlinear models with linear
systems tools. Dynamic mode decomposition (DMD), extended dynamic mode
decomposition (EDMD), kernel methods, and neural-network-based Koopman
embeddings have made Koopman-theoretic ideas computationally
accessible across fluid mechanics, robotics, control, and machine
learning
\cite{BOK:31,IM:05,MB-RM-IM:12,MOW-IGK-CWR:15,BL-JNK-SLB:18}. The
central premise is that, after a suitable change of observables,
nonlinear trajectories can be propagated, analyzed, or controlled
using finite-dimensional linear models.

However, finite-dimensional Koopman modeling has fundamental
limitations~\cite{SLB-MB-EK-JNK:22}.  Recent impossibility results
show~\cite{ZL-NO-EDS:25} that, under mild assumptions, continuous
immersions into finite-dimensional linear systems collapse distinct
omega-limit sets and therefore fail to provide globally injective
Koopman embeddings. This obstruction helps explain a recurring
empirical phenomenon: global Koopman models trained with smooth or
continuous dictionaries may blur basin boundaries, miss discontinuous
invariant coordinates, or produce long-term predictions that remain
accurate only on restricted regions of the state space.

A natural way to avoid this obstruction is to relax the requirement of
a single global finite-dimensional model and instead decompose the
state space into dynamically meaningful invariant regions. Once such
regions are known, one can learn a separate Koopman model on each
region, obtaining regional linear descriptions that respect the
distinct asymptotic behaviors of the nonlinear system. Another
advantage of this approach is that, since the state space will be
partitioned into forward invariant regions, no switching is needed to
transition between different models. Given an initial condition, the
model is determined by the partition and it is used for the entire
duration of the simulation. This idea has appeared in the literature
in several forms. Williams et al. \cite{MOW-IGK-CWR:15} use EDMD
eigenfunctions to identify the basins of attraction of the Duffing
oscillator and then applied EDMD separately on the resulting
subdomains. However, the procedure has misclassifications close to the
separatrix which compromises the quality of the learned models.
Nandanoori et al. \cite{SPN-SS-EY:20} read invariant sets off the
eigenvectors of a global Koopman operator and stitch regional Koopman
operators into a block-diagonal model via characteristic functions of
those sets, which must therefore be known in advance. We instead learn
the omega-limit partition directly from data and only then build the
regional Koopman models on it.

Regional Koopman constructions have also been used for data-driven
control, where separate linear models are learned in regions whose
dynamics are more amenable to finite-dimensional
approximation~\cite{GM-MC-XT-TM:19}.
%
%
For heteroclinic connections in plane Couette flow,~\cite{JP-RRK:19}
shows that each invariant solution supports its own Koopman expansion
of the state, valid only up to a crossover point, and that DMD fails
to recover these expansions when the data straddle such a point.  More
recently, Pan and Duraisamy \cite{SP-KD:24} study Koopman lifting and
reconstruction for systems with multiple invariant sets, emphasizing
the role of basin-dependent eigenfunctions and symmetry in
data-efficient learning. These works suggest that multiple Koopman
models are often the appropriate modeling object for systems with
multiple attractors or omega-limit sets, but they do not provide a
general data-driven procedure, with supporting theory, for learning
the underlying omega-limit partition itself.

Another related direction uses Koopman eigenfunctions to identify
basins of attraction and stability boundaries. In particular, Umathe
and Vaidya \cite{BU-UV:23} consider a system with an asymptotically
stable equilibrium and characterize the boundary of its region of
attraction as the union of the stable manifolds of the saddle-type
equilibria lying on this boundary; each such manifold is recovered as
the zero-level set of a Koopman eigenfunction associated with the
corresponding saddle.
%
%
Their goal is thus to recover the separatrix of a single region of
attraction, whereas ours is to recover the full omega-limit partition,
which distinguishes all asymptotic behaviors of the system at once.

Our contributions are as follows. We study autonomous dynamical
systems with possibly multiple omega-limit sets and define the
omega-limit partition, which groups together all initial conditions
whose trajectories share the same omega-limit set.
%
%
Our goal is to approximate it.  Invariant Koopman observables, i.e.,
eigenfunctions with unit eigenvalue, induce stationary partitions
(partitions of the state space into cells that trajectories never
leave)
%
%
through their level sets, and existing approaches exploit this fact by
decomposing the state space using a single invariant observable.  In
practice, however, numerical methods typically return several
approximate invariant observables, and it is unclear how to combine
the information they carry.  Moreover, the quality of these
approximations depends critically on the choice of dictionary, yet
there is little understanding of which dictionaries help and which
degrade the result. To address these issues, we study the relationship
between approximate invariant observables and the omega-limit
partition directly. This analysis yields a principled way to combine
multiple approximate invariants and clarifies the role of the
dictionary.

Specifically, we show that if the dictionary approximates the
indicator functions of the omega-limit cells well, then approximate
invariant observables never place points with different asymptotic
behavior in the same cell, except on a small set. However, a single omega-limit cell may be split into several fragments.
This result is the basis for
the design of our Omega-limit Partitions from Approximate invariants
Learned from data (OPAL) algorithm.  We compute approximate invariant
observables by solving a generalized eigenvalue problem and cluster
the data according to their values, which yields a partition that is
typically too fine. To address this, the remaining task is to decide
which fragments belong together.  We merge two fragments whenever the
data move frequently from one to the other in a single step
(\emph{transition-based merging}), which cannot happen between
distinct omega-limit cells, or whenever their invariant values nearly
coincide (\emph{value-based merging}), since true invariants are
constant on each cell. Finally, because even a handful of
misclassified points near a basin boundary can corrupt the Koopman
models later trained on each cell, we leave points close to the
recovered boundaries unassigned, trading coverage for confidence.
%
%

We then learn a separate Koopman model on each recovered cell.
Numerical examples on systems with multiple equilibria and multiple
limit cycles show that, with suitably chosen parameters, OPAL
classifies every retained point correctly while leaving less than
$2\%$ of the data unassigned. The resulting regional models improve
long-term prediction by an order of magnitude over a single global
model, whereas partitions that are merely nearly correct can produce
diverging regional models.

\section{Preliminaries}

We review\footnote{
  We use the following notation. Let $\real^n$ denote the
  $n$-dimensional Euclidean space. We let $[k]=\{ 1 , \dots, k\}$ for
  brevity. For a Hilbert space $\Hc$ with inner product
  $\langle \cdot, \cdot \rangle_{\Hc}$ and induced norm
  $\|\cdot\|_{\Hc}$, subspaces of observables are denoted by
  calligraphic letters (e.g., $\Sc$). We use $\text{span}(\cdot)$ to
  denote the linear span and $\text{dim}(\cdot)$ to denote
  dimension.} here basics on the
Koopman operator and the partition concepts used throughout the paper.

\subsection{The Koopman Operator}
\label{sec:koopman-prelims}
Following~\cite{AM-YS-IM:20}, consider a discrete-time dynamical
system on the state space $\mathcal{X} \subseteq \mathbb{R}^n$
described by a map $T: \mathcal{X} \to \mathcal{X}$:
\begin{equation}\label{eqn:sys_dynamics}
  x^+ = T(x), \quad x \in \mathcal{X}.
\end{equation}
The Koopman operator $\Kc: \mathcal{F} \to \mathcal{F}$ is an
infinite-dimensional linear operator that acts on a space
$\mathcal{F}$ of real-valued observables of the form
$ \psi: \mathcal{X} \to \mathbb{R}$ by composing them with the
dynamics, i.e., $(\Kc\psi)(x) = \psi(T(x))$. We assume the
function space $\mathcal{F}$ is closed under composition with $T$. We assume the space of observables $\mathcal{F}$ is
equipped with a Hilbert space structure, with inner product
$\langle \cdot, \cdot \rangle_{\mathcal{F}}$ and associated norm
$\|\cdot\|_{\mathcal{F}}$. A common choice is the $L_2$-space with
respect to a probability measure $\mu$ on $\mathcal{X}$. 

Given a dataset of snapshot pairs
$D_{\text{data}} = \{(x_i, y_i)\}_{i=1}^N$, throughout the paper, we
fix $\mu$ to be the empirical measure induced by the data
$D_{\text{data}}$ and take $\mathcal F = L_2(\mu)$. The goal of
data-driven approximation is to find a finite-dimensional operator
$K: \mathcal{S} \to \mathcal{S}$ that best represents the action of
$\mathcal{K}$ when restricted to $\mathcal{S}$. Extended Dynamic Mode
Decomposition (EDMD) \cite{lQL-FD-EMB-IGK:17} provides the optimal
approximation in the $L_2$-sense by orthogonally projecting the image
of the subspace back onto itself. Construct the data matrices
$\Psi(X), \Psi(Y) \in \mathbb{R}^{M \times s}$, where the $i$-th rows
are the evaluations of the dictionary functions $\Psi$ at $x_i$ and
$y_i$ respectively. The EDMD approximation $K_{\text{EDMD}}$ in the
basis $\Psi$ is the solution to the least-squares problem
\begin{align}
  \label{eq:LS-EDMD}
  \min_K \|\Psi(X) K - \Psi(Y)\|_F.
\end{align}

\subsection{Partitions and Stationary Partitions}

Given $\mathcal{X} \subseteq \real^n$, a
\emph{partition}~\cite{IPC-SVF-YGS:82}
%
%
is a family $\Pc=\{C_\alpha \subset \mathcal{X} \}_{\alpha\in A}$ of
disjoint nonempty subsets whose union is $\Xc$. Each $C_\alpha$ is an
\emph{element} (or \emph{cell}) of $\Pc$. A set $D\subseteq\Xc$ is a
\emph{$\Pc$-set} if it is a union of elements of~$\Pc$.  A partition
$\Pc_1$ is \emph{finer} than a partition $\Pc_2$, written
$\Pc_1\succeq\Pc_2$, if every element of $\Pc_2$ is a
$\Pc_1$-set. Equivalently, every $\Pc_2$-set is a $\Pc_1$-set. In this
case, $\Pc_2$ is \emph{coarser} than $\Pc_1$. We omit the partition
index when it causes no confusion.

Assume $\Xc$ is equipped with a $\sigma$-algebra $\Bc$ and a
probability measure $\mu$. 
A measurable partition~\cite{IPC-SVF-YGS:82}
$\mathcal P=\{C_\alpha\}_{\alpha\in A}$ is \emph{stationary} if every
element is forward invariant, i.e., $T(C_\alpha)\subseteq C_\alpha$,
for all $\alpha\in A$.  Note that stationary partitions are not
unique. One can combine cells of a stationary partition to obtain a
coarser stationary partition. For example, if $\{C_1, C_2\}$ is a
stationary partition, then $\{C_1 \cup C_2\}$ is also stationary. In
general, any partition that can be obtained by merging cells of a
stationary partition is itself stationary.

\subsection{Omega-Limit Sets and the Omega-Limit Partition}

We describe here basic properties of omega-limit sets, which play a
key role in the construction of the omega-limit partition, a unique
stationary partition of the state space $\Xc$ that captures the
asymptotic behavior of the system.
%
%
Given the dynamics~\eqref{eqn:sys_dynamics} and $x\in\Xc$, we write
$T^n$ for the $n$-fold composition of $T$ with itself (with $T^0$ the
identity map), so that $\{T^n(x)\}_{n\geq 0}$ is the \emph{trajectory}
of~\eqref{eqn:sys_dynamics} starting at~$x$. An \emph{accumulation
  point} $y\in\Xc$ of the trajectory starting at $x$ is such that
there exists a strictly increasing sequence of nonnegative integers
$\{n_j\}_{j\geq 1}$ with $\lim_{j\to\infty} T^{n_j}(x) = y$. The
\emph{omega-limit set}~\cite[Ch.~8]{SW:03} of $x$, denoted
$\omega(x)$, is the set of all accumulation points of the trajectory
starting at~$x$:
\begin{equation}
  \omega(x)
  :=
  \Big\{y\in\Xc: \exists \, \{n_j\}_{j\geq 1} \text{ s.t. }
    \lim_{j\to\infty} T^{n_j}(x) = y\Big\}. 
\end{equation}
Throughout the paper, we assume that $\Xc$ is a compact metric space
and $T$ is continuous. Under these assumptions, for each $x\in\Xc$,
the set $\omega(x)$ is nonempty, compact, and invariant in the sense
that $T(\omega(x)) = \omega(x)$.  Compactness of $\Xc$ guarantees that
every trajectory has at least one accumulation point, so
$\omega(x)\neq\emptyset$, and continuity of $T$ is what ensures the
invariance identity. Furthermore, since discarding finitely many terms
of a sequence does not change its accumulation points,
$\omega(T^k(x))=\omega(x)$ for every integer $k\geq 0$, without any
further assumption.
%

%


Equipped with the notion of an omega-limit set, one can define a relation
$\sim_\omega$ on $\Xc$ by $x\sim_\omega y \Longleftrightarrow
  \omega(x)=\omega(y)$.
The equivalence class of $x$ is
$R_\omega(x) := \{y\in\Xc:\omega(y)=\omega(x)\}$.  It is
straightforward to verify that $\sim_\omega$ is an equivalence
relation, so the equivalence classes $R_\omega(x)$ form a partition
of~$\Xc$. Indeed, the \emph{omega-limit partition} $\Pc_{\omega}$
of~\eqref{eqn:sys_dynamics} is the collection of equivalence classes: $\mathcal P_\omega
  :=
  \{R_\omega(x):x\in\Xc\}$.

The properties of omega-limit sets imply that $\mathcal P_\omega$ is
stationary. In contrast with the nonuniqueness of stationary
partitions discussed above, $\Pc_\omega$ is canonically determined by
the dynamics: it is the coarsest partition of $\Xc$ for which points
in the same cell share the same omega-limit set.
%
%

\section{Problem Statement}
Consider a discrete-time autonomous dynamical system of the
form~\eqref{eqn:sys_dynamics}.  Unless stated otherwise, $\Xc$ is a
compact metric state space, $T:\Xc\to\Xc$ is continuous, and $\Xc$ is
forward invariant under the dynamics. We assume that the function
space $\Fc = L_2(\mu)$ as defined in \ref{sec:koopman-prelims} is closed under composition with $T$. 

{The system~\eqref{eqn:sys_dynamics} may have multiple omega-limit
sets, corresponding to distinct asymptotic behaviors. This
multiplicity of asymptotic behaviors is a fundamental obstruction to
global finite-dimensional Koopman modeling with continuous
observables.}
%
%
To tackle this, our goal is to learn the omega-limit partition
$\mathcal P_\omega$ from data. The rationale for this is that, once
the cells are known, one can learn Koopman models on each cell, and
these collectively capture the different asymptotic behaviors of the
system. 
%

We assume access to snapshot pairs $\{(x_i, y_i)\}_{i=1}^N$ with
$y_i = T(x_i)$, together with a dictionary of candidate observables
$\Psi = \{\psi_1, \dots, \psi_s: \mathcal{X} \rightarrow \real\}$
spanning a subspace $\Sc$. Our objective is to extract a partition
$\hat{\mathcal P}$ from this data, built from approximate Koopman
invariants, that approximates the true omega-limit
partition~$\mathcal P_\omega$. In constructing this approximation, we want to obtain a partition $\widehat{\Pc}$ as coarse as possible so that fewer Koopman models are needed, while making sure that all distinct asymptotic behaviors are captured. 
%
%

{To do so, we first establish the connection between invariant
observables and the omega-limit partition,
cf. Section~\ref{sec:invariant_observables}. We then build on this
understanding to study how approximate invariants can be employed to
approximate the true omega-limit partition,
cf. Section~\ref{sec:approximate_invariants}. Finally, we put
everything together in an algorithmic procedure, which we term OPAL
(Omega-limit Partitions from Approximate invariants Learned from
data), to learn the omega-limit partition from data, cf.
Section~\ref{sec:post_processing_partition}.}
%
%


\section{Invariant Observable
  Partitions}\label{sec:invariant_observables}

Our first step is to establish the connection between invariant
observables and the omega-limit partition.

\begin{definition}[Invariant observables]
  We refer to the elements of the eigenspace $E_{\text{inv}}(\Kc):=
    \{f\in\Fc:\Kc f=f\}$ of
  the Koopman operator corresponding to the eigenvalue $\lambda=1$ as
  invariant observables of \eqref{eqn:sys_dynamics}. \oprocend
\end{definition}
\smallskip

The following result describes how invariant observables induce
stationary partitions through their level sets.

\begin{proposition}\longthmtitle{Partitions induced by invariant
    observables}
  Let $f\in E_{\text{inv}}(\Kc)$. For $c\in\R$, define the level set
  \begin{equation}
    C_c^f := \{x\in\Xc:f(x)=c\}.
  \end{equation}
  %
  %
  Then each nonempty level set $C_c^f$ is forward
  invariant. Consequently, the collection of nonempty level sets
  $ \{C_c^f : c \in \R\} $ of $f$ is a stationary partition of $\Xc$.
\end{proposition}
\begin{proof}
  Let $c\in\R$ and suppose $x\in C_c^f$. Then $f(x)=c$. Since $f$ is
  invariant, $f \circ T (x)=f(x)=c$. Therefore $T(x)\in C_c^f$ and
  $T(C_c^f)\subseteq C_c^f$. The nonempty level sets of any function
  are pairwise disjoint and cover $\Xc$. Since each nonempty level set
  is forward invariant and each $C_c^f$ is a measurable set, the
  collection of nonempty level sets is a stationary partition.
\end{proof}

Families of invariant observables induce partitions through their
joint level sets, as we detail next.

\begin{proposition}\longthmtitle{Partition induced by a family of
    invariant observables}
  Let $\mathcal A\subseteq E_{\text{inv}}(\Kc)$ be a family of invariant
  observables. Define the equivalence relation
  \begin{equation}
    x\sim_{\mathcal A} y
    \; \Longleftrightarrow \;
    f(x)=f(y)\quad\text{for every } f\in\mathcal A.
  \end{equation}
  Let
  $\mathcal P_{\mathcal A}:=\{[x]_{\mathcal A}:x\in\Xc\}$ be the partition induced, where
  $ [x]_{\mathcal A}$ denotes the equivalence class.  Then,
  $\mathcal P_{\mathcal A}$ is stationary.
\end{proposition}
\begin{proof}
  The relation $\sim_{\mathcal A}$ is an equivalence relation, so its
  equivalence classes form a partition of $\Xc$. Fix a cell
  $[x]_{\mathcal A}$ and let $y\in[x]_{\mathcal A}$. Then, for every
  $f\in\mathcal A$, we have $f(y)=f(x).$ Since each $f\in\mathcal A$
  is invariant: $f \circ T (y)=f(y)=f(x)$.
  Thus $T(y)\in[x]_{\mathcal A}$. Hence each cell is forward
  invariant, so $\mathcal P_{\mathcal A}$ is stationary.
\end{proof}

Since $[x]_{\mathcal A}=\bigcap_{f\in\mathcal A} C^f_{f(x)}$, the
partition $\mathcal P_{\mathcal A}$ is the common refinement of the
partitions of the form $\{C^f_c: c\in\R\}$ induced by the individual
observables $f\in\mathcal A$: it is finer than each of them, and it is
the coarsest such partition.

\section{Approximate invariant observable
  partitions}\label{sec:approximate_invariants}

Here, we turn our attention to the data-driven setting.  We show how
one can obtain a partition from approximate invariant observables
learned from data that, outside what we call an exception set (whose
size is tunable), is finer than the true omega-limit
partition~$\mathcal P_\omega$.

%

Given $\epsilon > 0$, define the set of $\epsilon$-invariant
observables in the subspace~$\Sc$ by
$
  E_{\text{inv}}^{\epsilon}(\Sc) := \{f\in\Sc: \| \Kc f - f \| \leq
  \epsilon \|\Kc f\| \} .
$
The partition induced is $\mathcal P_{\epsilon} := \{[x]_{\epsilon}: x \in \Xc\}$, where
\begin{align}\label{eq:P_epsilon}
  [x]_{\epsilon} := \{y \in \Xc: f(y) = f(x) , \;
  \forall f \in E_{\text{inv}}^{\epsilon}(\Sc)\}.
\end{align}
Unlike the exact invariant-observable partitions of
Section~\ref{sec:invariant_observables}, $\Pc_{\epsilon}$ need not be
stationary because the observables in $E_{\text{inv}}^{\epsilon}(\Sc)$
are only approximately invariant. The next result gives conditions
under which $\Pc_{\epsilon}$ is a finer partition than
$\mathcal P_\omega$ (i.e., $\Pc_{\epsilon} \succeq \mathcal P_\omega$)
outside a small exception set which we quantify precisely in the
statement.
%
%
%

\begin{theorem}\longthmtitle{Approximate invariant observables
    separate omega-limit cells outside a small set}\label{thm:main_result}
  Assume the omega-limit partition $P_\omega$
  of~\eqref{eqn:sys_dynamics} has a finite number of cells
  $ \{C_1,\dots,C_q, \dots, C_t\}$, is measurable and arranged so that
  $\mu(C_i)>0$, for $i \in [q]$ and $\mu(C_i)=0$ for
  $i \in [t] \setminus [q]$.
  %
  %
  Let $\chi_i:=\mathbf{1}_{C_i}$ be the indicator
  function of $C_i$. Assume that for each $i \in [q]$, $\exists$ 
  $s_i\in \Sc$, $\delta_i>0$, $\beta_i \geq 0$ and $\epsilon>0$ s.t.
  \begin{subequations}\label{eq:conditions-dictionary}
    \begin{align}
      \|{s_i-\chi_i}\|_{L_2(\mu)}
      & \leq \delta_i,\
        \|{\Kc(s_i-\chi_i)}\|_{L_2(\mu)} \leq
        \beta_i\delta_i, \label{eqn:approx_indicator}
      \\  
      \beta_i \delta_i
      &< \sqrt{\mu(C_i)}, \,\,
        \frac{(\beta_i+1)\delta_i}
        {\sqrt{\mu(C_i)}-\beta_i\delta_i}
        \leq \epsilon. \label{eq:epsilon_condition}
    \end{align}
  \end{subequations}
  %
  %
  Then the functions $s_i$ are $\epsilon$-invariant observables, i.e.,
  \begin{equation}
    s_i\in E_{\text{inv}}^{\epsilon}(\Sc) \text{ for every } i \in [q].
  \end{equation}
  Furthermore, let $Z:=\bigcup_{i=q+1}^{r} C_i$ be the union of the
  null cells and, for any $\eta\in(0,1/2)$, define the exception set
  $B_\eta := Z \cup \bigcup_{i=1}^q B_i(\eta)$, where
  \begin{equation}
    B_i(\eta) :=\left\{x\in \Xc:|s_i(x)-\chi_i(x)|>\eta\right\}.
  \end{equation}
  Then the measure of the exception set satisfies $\mu(B_\eta) \leq \frac{1}{\eta^2} \sum_{i=1}^q \delta_i^2$. 
  Finally, for all $x,y\in \Xc\setminus B_\eta$, we have
  \begin{equation}\label{eqn:refinement_result}
    f(x)=f(y), \; \forall f\in E_{\text{inv}}^{\epsilon}(\Sc) \,\,
    \Longrightarrow \,\, \omega(x)=\omega(y). 
  \end{equation}
  Equivalently, outside $B_\eta$, the partition $\Pc_{\epsilon}$
  induced by $E_{\text{inv}}^{\epsilon}(\Sc)$ is finer than the
  omega-limit partition~$\mathcal P_\omega$.
\end{theorem}
\begin{proof}
  We split the proof into three parts.

  \noindent \emph{Part 1:} (Each approximating indicator is an
  approximate invariant observable)
  Let $i \in [q]$. Since $\chi_i$ is an exact invariant observable,
  $K\chi_i=\chi_i$, we have $ \Kc s_i-s_i
    = \Kc(s_i-\chi_i) - (s_i-\chi_i)$.
  By the triangle inequality and \eqref{eqn:approx_indicator},
  \begin{align*}
    \|{\Kc s_i-s_i}\|_{L_2(\mu)}
    &\leq \|{\Kc(s_i-\chi_i)}\|_{L_2(\mu)} + \|{s_i-\chi_i}\|_{L_2(\mu)}
    \\
    &\leq \beta_i\delta_i + \delta_i = (\beta_i+1)\delta_i.
  \end{align*}
  Next, using $\Kc\chi_i=\chi_i \implies \| \Kc\chi_i\| = \|\chi_i \| = \sqrt{\mu(C_i)}$,
  \begin{align*}
    \|{\Kc s_i}\|_{L_2(\mu)}
    &\geq
      \|{\Kc\chi_i}\|_{L_2(\mu)}
      -
      \|{\Kc(s_i-\chi_i)}\|_{L_2(\mu)} \\
    &\geq
      \sqrt{\mu(C_i)}-
      \beta_i\delta_i.
  \end{align*}
  The right-hand side is positive by
  \eqref{eq:epsilon_condition}. Hence
  \begin{align*}
    \frac{\|{\Kc s_i-s_i}\|_{L_2(\mu)}}
    {\|{\Kc s_i}\|_{L_2(\mu)}}
    &\leq
      \frac{(\beta_i+1)\delta_i} {\sqrt{\mu(C_i)}-\beta_i\delta_i} \leq \epsilon
  \end{align*}
  by \eqref{eq:epsilon_condition}. Therefore
  $s_i\in E_{\text{inv}}^{\epsilon}(\Sc)$.
  
  \noindent \emph{Part 2:} (The exception set has small measure)
  By Markov's inequality~\cite[Thm.~6.17]{GBF:99},
  \begin{align*}
    \mu(\left\{x\in \Xc:|s_i(x)-\chi_i(x)|>\eta\right\})
    &\leq \frac{\|{s_i-\chi_i}\|_{L_2(\mu)}^2}{\eta^2}.
  \end{align*}
  Using~\eqref{eqn:approx_indicator}, we deduce that
  $\mu(B_i(\eta)) \leq \frac{\delta_i^2}{\eta^2}$.  Moreover,
  $\mu(Z)\leq\sum_{i=q+1}^{t}\mu(C_i)=0$.  Since
  $B_\eta=Z\cup\bigcup_{i=1}^q B_i(\eta)$, 
  \[
    \implies \,\, \mu(B_\eta) \leq \mu(Z) + \sum_{i=1}^q \mu(B_i(\eta)) \leq
    \frac{1}{\eta^2}\sum_{i=1}^q \delta_i^2.
  \]

  \noindent \emph{Part 3:} (Separation outside the exception set)
  Let $x,y\in \Xc\setminus B_\eta$. Since $Z\subseteq B_\eta$,
  neither $x$ nor $y$ belongs to a null cell, so
  $x,y\in\bigcup_{i=1}^q C_i$. We prove the contrapositive. Suppose
  $\omega(x)\neq\omega(y)$, so $x$ and $y$ lie in different
  omega-limit cells. This means that there is an index $i\in[q]$
  such that $x\in C_i$ and $y\notin C_i$; in particular, the
  approximating observable $s_i$ exists. Since $x\notin B_\eta$, in
  particular $x\notin B_i(\eta)$, and using $\chi_i(x)=1$, we get
  \begin{equation}
    |s_i(x)-\chi_i(x)|\leq \eta \implies s_i(x)\geq 1-\eta.
  \end{equation}
  Similarly, since $y\notin B_i(\eta)$ and $\chi_i(y)=0$,
  \begin{equation}
    |s_i(y)-\chi_i(y)|\leq \eta \implies s_i(y)\leq \eta.
  \end{equation}
  Since $\eta<1/2$, we have $1-\eta > \eta$, and therefore
  $s_i(x)\neq s_i(y)$.  Since, from Step~1,
  $s_i\in E_{\text{inv}}^{\epsilon}(\Sc)$, this proves the
  contrapositive, i.e., there exists an observable in
  $E_{\text{inv}}^{\epsilon}(\Sc)$ that separates $x$ and $y$.
\end{proof}

{Basically, Theorem~\ref{thm:main_result} says that if the
approximation error \eqref{eq:conditions-dictionary} is small enough
relative to the measure of the cell, then the approximate invariants
form a refinement of the omega-limit partition outside a small
exception set \eqref{eqn:refinement_result}.
As a consequence of Theorem~\ref{thm:main_result}, the cells of
$\Pc_{\epsilon}$ can be merged to recover the coarser omega-limit
partition.}
%
%
Several remarks are in order to interpret the implications of this
result.

%
%

\begin{remark}\longthmtitle{Choice of dictionary functions}
  %
  %
  The result hinges upon the existence of functions $s_i\in\Sc$ that
  approximate the indicator functions $\chi_i$ of the omega-limit
  cells $C_i$ reasonably well, as quantified by the conditions in
  \eqref{eq:conditions-dictionary}.  If these conditions are met, then
  the subspace $\Sc$ contains approximate invariant observables that
  nearly separate those cells. This suggests using dictionaries rich
  enough to approximate piecewise-constant functions, since the target
  indicators are discontinuous across cell boundaries. \oprocend
\end{remark}

\begin{remark}\longthmtitle{Interpretation of the approximation
    conditions} 
  %
  The smallness
  condition is equivalently written as
  \begin{equation*}
    \frac{(\beta_i+1)\delta_i}
    {\sqrt{\mu(C_i)}-\beta_i\delta_i}
    \leq \epsilon \iff     \delta_i
    \leq
    \frac{\epsilon\sqrt{\mu(C_i)}}{\beta_i+1+\epsilon\beta_i}.    
  \end{equation*}
  Thus, for small $\epsilon$, the dictionary must approximate the
  indicator of each omega-limit cell accurately. The approximation
  requirement is stricter for small-measure cells. The measure of the
  exception set can be made arbitrarily small by improving the
  approximation quality of the indicators. \oprocend
\end{remark}

\begin{remark}[Regarding the boundedness of $\Kc$]
  If the Koopman operator $\Kc$ is well defined on $L_2(\mu)$, then one
  can restrict it to the finite-dimensional subspace
  $\Sc + \text{span}\{\chi_i\}$,
  where it is also bounded. The condition
  $\|{\Kc(s_i-\chi_i)}\|_{L_2(\mu)} \leq \beta_i\delta_i$ for $i \in [q]$ can be
  satisfied by choosing $\beta_i$ to be the corresponding operator
  norm.
  %
  %
  \oprocend
\end{remark}

\section{Construction of the approximate
  partition}\label{sec:post_processing_partition}
%
%

This section describes our approach to construct the approximation of
the true omega-limit partition~$\mathcal P_\omega$. We begin by
outlining how to compute a dictionary of approximate invariant
observables from data, cf. Section~\ref{sec:approx-inv-obs}.
Next, we describe how to combine the
cells of $\Pc_\epsilon$  to recover the coarser stationary
partition~$\Pc_{\omega}$. The key point in doing this is that the
elements of $E_{\text{inv}}^\epsilon(\Sc)$ are only approximately
invariant. Consequently, their exact level sets are not expected to be
forward invariant. However, their values should change only mildly
along one-step transitions. {Our post-processing algorithm exploits
this distinction: it first constructs an intermediate empirical
partition from the values of the approximate invariants,
cf. Section~\ref{sec:partition-clustering}, and then merges cells that
are close in invariant-observable value space or connected by one-step transitions, cf.  Section~\ref{sec:merging}.} Finally, we robustify the resulting cells by discarding points near
their boundaries, cf. Section~\ref{sec:tau_robust}. This significantly boosts the classification accuracy of the retained cells.

%
%
%

\subsection{Computing the approximate invariant
  observables}\label{sec:approx-inv-obs}
We first compute the approximate invariant observables
$E_{\text{inv}}^\epsilon(\Sc)$ from data. Given the dictionary
$\Psi=\{\psi_1,\ldots,\psi_s\}$ generating the subspace~$\Sc$, a
candidate observable is
\begin{equation}
  f_c(x)=\Psi(x)c:=\sum_{j=1}^s c_j\psi_j(x),
  \quad c\in\R^s.
\end{equation}
Given snapshot pairs $D_{\text{data}} = \{(x_i,y_i)\}_{i=1}^N$, define
\begin{equation}
  \Psi(X)_{ij}:=\psi_j(x_i),
  \qquad
  \Psi(Y)_{ij}:=\psi_j(y_i).
\end{equation}
Set $M_\Psi := (\Psi(Y)-\Psi(X))^{\top} (\Psi(Y)-\Psi(X))$ and
$N_\Psi := \Psi(Y)^{\top}\Psi(Y)$. Then
\begin{equation}
  \min_{\substack{f_c \in \Sc \\ \Kc f_c \neq 0}} \frac{\|\Kc f_c-f_c\|^2_{L_2(\mu)}}{\|\Kc f_c\|^2_{L_2(\mu)}} = 
  \min_{\substack{c \in \real^s \\ \Psi(Y)c \neq 0}} \frac{c^{\top}M_\Psi c}{c^{\top}N_\Psi c}.
\end{equation}
%
%
Thus, approximate invariants are obtained from the solution of a
generalized eigenvalue problem:
\begin{equation}
  M_\Psi c = \lambda N_\Psi c.
  \label{eq:empirical_GEP}
\end{equation}
%
%
Let $c_1,\ldots,c_r$ be the retained eigenvectors satisfying
$\lambda\leq \epsilon^2$ and define $f_j(x):=\Psi(x)c_j$, for
$j\in[r]$. Then we have
\begin{equation}
  \{f_1,\ldots,f_r\} \in E_{\text{inv}}^\epsilon(\Sc)
\end{equation}
%
%

\subsection{Partition by clustering invariant
  observables}\label{sec:partition-clustering}
%
%

The partition $\Pc_\epsilon$ is defined by equality of all functions
in $E_{\text{inv}}^\epsilon(\Sc)$, cf.~\eqref{eq:P_epsilon}. Our goal
is to create a coarser partition to better
approximate~$\Pc_\omega$. To do this, we begin by clustering the data
points according to the values of the approximate invariant
observables $f_1,\ldots,f_r$. For each $j\in[r]$ and $i \in [N]$, we evaluate $f_j$ on the data $u_i^j:=f_j(x_i)$. 
For a given integer $k_j$, we cluster the scalar values
$\{u_i^j\}_{i=1}^N$ into sets
$\Ic_1^j,\ldots,\Ic_{k_j}^j\subset\R$. The numbers $k_j$ may be
chosen larger than the expected number of omega-limit cells, since
subsequent steps merge the resulting over-refinement.

For each data point $x_i$, define the scalar cluster label
\begin{equation}
  a_j(i)=\ell
  \quad\Longleftrightarrow\quad
  f_j(x_i)\in\Ic_\ell^j.
\end{equation}
The vector of scalar labels of the data point $x_i$ is given by
\begin{equation}
  a(x_i):=(a_1(i),\ldots,a_r(i))
  \in [k_1]\times\cdots\times[k_r]
\end{equation}
%

For a new point $x \in \Xc$, we use the same notation $a(x)$ for the
label vector obtained by evaluating each $f_j(x)$ and assigning it to
the appropriate cluster $\Ic_\ell^j$.

Note that we choose to not explicitly enumerate all
$\prod_{j=1}^r k_j$ possible label vectors as this would exponentially
increase the number of cells. Instead, we keep only the label vectors
that occur in the dataset, $\mathfrak A :=\{a(x_i):i\in[N]\}$.
This has a disadvantage that for a new point $x\in\Xc$, the label
vector $a(x)$ may not belong to $\mathfrak A$. We assume that we
have enough data points so that the measure of the set of points
with label vectors not in $\mathfrak A$ is negligible.  

For $\alpha\in\mathfrak A$, define the cell $C_\alpha :=\{x \in \Xc : a(x)=\alpha\}$. With a slight abuse of notation, we refer to the collection
$\Pc_{\mathfrak{A}}:=\{C_\alpha:\alpha\in\mathfrak A\}$
%
%
as the ``partition'' generated by sampled data. Recall that
$\Pc_\epsilon$ is the partition induced by equality of all functions
in $E_{\text{inv}}^\epsilon(\Sc)$,
cf.~\eqref{eq:P_epsilon}. $\Pc_\epsilon$ is finer than
$\Pc_\mathfrak{A}$, because the latter is defined by clustering the
values of the functions in $E_{\text{inv}}^\epsilon(\Sc)$.  {Also, we
note that $\Pc_\mathfrak{A}$ need not cover the entire set $\Xc$
and hence is not a
true partition of~$\Xc$.}

\subsection{Merging the partition}\label{sec:merging}

%
%

For each occupied cell, let
$ \widehat C_\alpha = \{x_i \in D_{\text{data}} : a(x_i)=\alpha\}$ be
the corresponding empirical cell. Note that $\widehat C_\alpha$ is
different from $C_{\alpha}$ because the former is a collection of the
data points only. Define its empirical weight and embedding centroid
by
\begin{equation}
  w_\alpha:=\frac{|\widehat C_\alpha|}{N}, \,\,     m_\alpha
  :=\frac{1}{| \widehat C_\alpha|} \sum_{x_i\in
    \widehat C_\alpha}F(x_i), 
  \label{eq:centroid_definition}
\end{equation}
where $F(x)=(f_1(x),\ldots,f_r(x))$ is the embedding of $x$ in the
space of approximate invariants. The weight $w_\alpha$ is
measure of $C_\alpha$ under the empirical measure $L_2(\mu)$. 

For a given occupancy threshold $\tau_{\text{occ}}$, cells satisfying
$w_\alpha<\tau_{\text{occ}}$ are treated as low-occupancy cells and
discarded. This is important because low occupancy cells are likely to
be spurious and not representative of the underlying
dynamics. Moreover, retaining them creates problems with the merging
procedure described below, cf. Remark~\ref{remark:occupancy}.

Let $\mathfrak A_{\text{occ}}$ be the set of occupied labels
corresponding to the retained cells, and let
$\Pc_{\mathfrak A_{\text{occ}}} =
\{C_\alpha:w_\alpha\geq\tau_{\text{occ}}\}$ be the corresponding
retained partition. Next, we merge the cells of
$\Pc_\mathfrak{A_{\text{occ}}}$ to create a coarser partition. In
order to do this, we construct a merge graph $\Gc_{\text{merge}}=(\mathfrak A_{\text{occ}},\Ec_{\text{merge}})$ whose nodes are the retained fine cells.  We describe how to add edges
to the merge graph based on transition and value criteria. The final
partition is obtained by taking the connected components of this
graph. In order to visualize the procedure as we present it, we
provide Figure~\ref{fig:merge_schematic} to illustrate the two-stage
merge used in our implementation: transition-based edges on fine
cells, followed by value-based edges between centroids of the
transition-merged components.
%
%
\subsubsection{Transition-based merging}
Because we seek to approximate the $\omega$-limit partition, each
one-step pair $(x_i,y_i)$ is expected to remain within a single
omega-limit cell. For $\alpha,\beta\in\mathfrak A_{\text{occ}}$,
define the transition count
\begin{equation}
  n_{\alpha\beta} := | \, \{i\in[N]:x_i\in  C_\alpha, \ y_i\in
  C_\beta\} |. 
\end{equation}
%
%
%
%
The empirical transition probability from $C_\alpha$ to
$ C_\beta$ is then
$p_{\alpha\beta} :={n_{\alpha\beta}}/{n_\alpha}$, where
$n_\alpha:=|\widehat C_\alpha|$.
If $C_\alpha$ and $C_\beta$ were cells of a stationary partition, then
transitions from one to the other would not occur. Because the
partition $\Pc_{\mathfrak A_{\text{occ}}}$ is built from approximate
invariants, frequent transitions between two cells indicate that the
cells are fragments of the same omega-limit cell. Hence, we add an
edge between $\alpha,\beta$ when
$
\boxed{
\max\{p_{\alpha\beta},p_{\beta\alpha}\}
  \geq \tau_{\text{tr}},}
$
where $\tau_{\text{tr}}>0$ is a transition threshold.

\begin{remark}\longthmtitle{Regarding $\tau_{\text{occ}}$}\label{remark:occupancy}
  The occupancy threshold guards the transition criterion: since
  $p_{\alpha\beta}=n_{\alpha\beta}/n_\alpha$, a cell with very few
  samples can exceed $\tau_{\mathrm{tr}}$ after only one or two
  transitions. Low-occupancy cells that straddle a separatrix
  therefore transition into several basins and can act as
  bridges that merge distinct omega-limit cells. \oprocend
\end{remark}

\begin{figure}[h]
  \centering
  \includegraphics[width=0.9\linewidth]{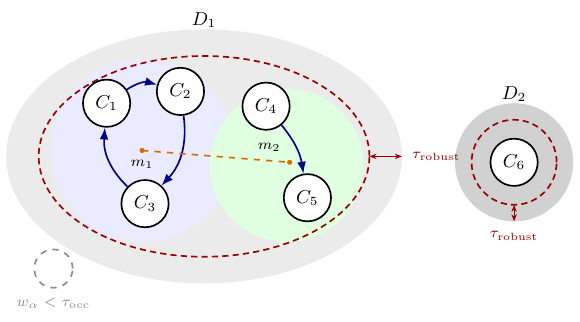}
  \caption{Schematic of the two-stage merge procedure. Fine cells
    $C_1, \dots, C_6$ are first formed by
    clustering approximate invariant observables and dropping low occupancy cells based on $\tau_{\text{occ}}$. The obtained approximate partition is $\{ \widehat D_1, \widehat D_2 \}$.}
  \label{fig:merge_schematic}
\end{figure}

\subsubsection{Value-based merging}
%
%
We also add edges based on the values of the approximate
invariants. If the centroids of two cells are close in
invariant-observable value space, then the two cells may represent the
same omega-limit cell. Indeed, approximate invariants are expected to
be nearly constant on each omega-limit cell, so points in the same
omega-limit cell should have similar invariant-observable
values. Consequently, we add an edge between $\alpha$, $\beta \in \mathfrak{A}_{\text{occ}}$
whenever
$\boxed{\|m_\alpha-m_\beta\|_2
  \leq \tau_{\text{val}}},$
where $m_{\alpha}$ and $m_{\beta}$ are the centroids defined in
\eqref{eq:centroid_definition} and $\tau_{\text{val}}$ is a
user-defined threshold.

\subsection{Boosting classification accuracy}
\label{sec:tau_robust}

Once all edges have been added, we let
$\mathfrak M_1,\ldots,\mathfrak M_L$ be the connected components of
$\Gc_{\text{merge}}$ and define
\begin{equation}
  D_\ell
  :=\bigcup_{\alpha\in\mathfrak M_\ell} C_\alpha,
  \qquad \ell\in[L].
  \label{eq:merged_cell}
\end{equation}
Before returning these cells as the desired approximate partition, we introduce a
parameter $\tau_{\text{robust}}$ that trades off the set of points
the approximate partition covers for confidence in the cells it
returns. We shrink the region covered by the partition to ensure that, outside the exception set, the
partition matches the true omega-limit partition.

The need for this parameter comes from the fact that the scalar
clusters $\Ic_\ell^j$ are learned from data. Their endpoints only
approximate the true boundaries between omega-limit cells, so
misclassifications concentrate in a thin band around those endpoints
at which the merged cell changes, while points deep inside a merged
cell are classified correctly. Since each $D_\ell$ is a union of
cells $C_\alpha$, it induces new scalar clusters along each
coordinate. For $\alpha\in\mathfrak M_\ell$ and $j\in[r]$, let
$\Jc_\alpha^j\subset\R$ be the union of the consecutive scalar
clusters $\Ic_m^j$ containing $\Ic_{\alpha_j}^j$ that are obtained by
varying only the $j$-th entry of $\alpha$ while remaining in
$\mathfrak M_\ell$. In words, $\Jc_\alpha^j$ is the interval of
values of $f_j$ over which the merged cell of $\alpha$ does not
change when the other coordinates are held fixed. We fix
$\tau_{\text{robust}}\in[0,1)$ and let
$\widetilde\Jc_\alpha^j\subset\Jc_\alpha^j$ be the interval obtained
by removing a fraction $\tau_{\text{robust}}$ of the length of
$\Jc_\alpha^j$ from each of its endpoints at which the merged cell
changes. The robustified merged cell is then
\begin{equation}
  \widetilde D_\ell
  :=\bigcup_{\alpha\in\mathfrak M_\ell}
  \{x\in C_\alpha : f_j(x)\in\widetilde\Jc_\alpha^j,\ j\in[r]\}
  \subset D_\ell.
  \label{eq:robust_cell}
\end{equation}
Points in $D_\ell\setminus\widetilde D_\ell$ are not assigned to any
cell. These are the points we are not confident about. Setting
$\tau_{\text{robust}}=0$ recovers $\widetilde D_\ell=D_\ell$. The resulting approximate partition is
\begin{equation}
  \widehat \Pc
  :=\{\widetilde D_1,\ldots, \widetilde D_L\}.
  \label{eq:final_partition}
\end{equation}
%
%
{Overall, the procedure starts from $\Pc_{\epsilon}$ and returns
$\widehat \Pc$, where $\Pc_{\epsilon}$ is finer than $\widehat{\Pc}$.}
This coarsening is designed to preserve distinctions supported by the
joint invariant-observable embedding while removing spurious
fragmentation caused by approximate invariance, numerical noise, and
over-clustering. The entire procedure is summarized in
Algorithm~\ref{alg:post_process}.

%
%

\begin{algorithm}[t]
  \caption{OPAL: Approximating the $\omega$-limit partition}
  \label{alg:post_process}
  \begin{algorithmic}[1]
    \REQUIRE Snapshot pairs $\{(x_i,y_i)\}_{i=1}^N$, dictionary
    $\Psi=\{\psi_1,\ldots,\psi_s\}$, tolerance $\epsilon$, scalar
    cluster counts $k_1,\ldots,k_r$, thresholds $\tau_{\text{occ}}$,
    $\tau_{\text{tr}}$, $\tau_{\text{val}}$, margin
    $\tau_{\text{robust}}$.  \STATE Form $M_\Psi$,
    $N_\Psi$ and solve the GEP \eqref{eq:empirical_GEP}.

    \STATE Select eigenvectors $c_1,\ldots,c_r$ satisfying
    $\lambda\leq\epsilon^2$ and define $f_j(x)=\Psi(x)c_j$ for
    $j\in[r]$.

    \STATE For each $j\in[r]$, cluster $\{f_j(x_i)\}_{i=1}^N$ into
    $\Ic_1^j,\ldots,\Ic_{k_j}^j$.  \STATE For each $x_i$, compute the
    multi-index $a(x_i)=(a_1(i),\ldots,a_r(i))$ induced by the scalar
    clusters.  \STATE Form only the occupied fine cells
    $\widehat C_\alpha=\{x_i:a(i)=\alpha\}$ for
    $\alpha\in\mathfrak A=\{a(i):i\in[N]\}$.  \STATE Compute cell
    weights $w_\alpha=|\widehat C_\alpha|/N$ and embedding centroid
    $m_\alpha$. Retain cells with $w_\alpha\geq\tau_{\text{occ}}$.

    \STATE Build a merge graph with one node per retained fine cell.

    \STATE Add edges $(\alpha,\beta)$ based on
    $\tau_{\text{tr}}$ and $\tau_{\text{val}}$.

    \STATE Compute the connected components
    $\mathfrak M_1,\ldots,\mathfrak M_L$ and the merged cells
    $D_1,\ldots,D_L$ in \eqref{eq:merged_cell}.

    \STATE Form the robustified
    cells $\widetilde D_1,\ldots,\widetilde D_L$ in
    \eqref{eq:robust_cell}.

    \STATE Return the approximate partition
    $\widehat\Pc=\{\widetilde D_1,\ldots,\widetilde D_L\}$.
  \end{algorithmic}
\end{algorithm}

\section{Simulation Results}\label{sec:numerical_examples}
We demonstrate the proposed approach on two systems with multiple
asymptotic
behaviors\footnote{\url{https://github.com/dhruvshah0208/OPAL-Koopman}.}. All
simulations are performed in Python 3.11.4 on a machine with an Apple
M1 Pro chip and 16 GB of RAM. The simulations run in less than a
minute including the generalized eigenvalue problem and the
post-processing steps.

\begin{table}[htb!]
  \centering
  \small
  \begin{tabular}{@{}lcc@{}}
    \toprule
    & Duffing & Bistable osc. \\
    \midrule
    Trajectories $\times$ time steps & $500\times 100$ & $500\times 50$ \\
    Initial conditions & $[-2,2]^2$ & $[0,4]\times[0,2\pi]$ \\
    Snapshot pairs $N$ & $5\times 10^4$ & $2.5\times 10^4$ \\
    \midrule
    \multicolumn{3}{@{}l}{\emph{Partition (Algorithm~\ref{alg:post_process})}} \\
    RBF width $\gamma$ & $30$ & $5$ \\
    Polynomial degree & $0$ & $3$ \\
    Tolerance $\epsilon$ & $10^{-4}$ & $10^{-4}$ \\
    Cluster counts $(k_1,k_2)$ & $(2,4)$ & $(1,4)$ \\
    $\tau_{\mathrm{occ}}$ & $2\times 10^{-3}$ & $8\times 10^{-4}$ \\
    $\tau_{\mathrm{tr}}$ & $5\times 10^{-3}$ & $0.1$ \\
    $\tau_{\mathrm{val}}$ & $0$ & $6\times 10^{-3}$ \\
    $\tau_{\mathrm{robust}}$ & $0.05$ & $0.2$ \\
    Training snapshots dropped & $1.8\%$ & $0.84\%$ \\
    \midrule
    \multicolumn{3}{@{}l}{\emph{Regional and vanilla EDMD}} \\
    RBF centers & $100$ & $100$ \\
    RBF width $\gamma$ & $1$ & $1$ \\
    Polynomial degree & $6$ & $6$ \\
    Prediction horizon & $400$ & $300$ \\
    Validation trajectories & $47$ & $97$ \\
    \bottomrule
  \end{tabular}
  \caption{Data, partition, and EDMD parameters for the two
    examples. {Validation trajectories are those with correctly
    classified initial conditions used for the rollout errors.}}
  \label{tab:parameters}
\end{table}

\begin{figure*}[htb!]
  \centering
  \begin{subfigure}[t]{0.49\linewidth}
    \centering
    \includegraphics[width=.9\linewidth]{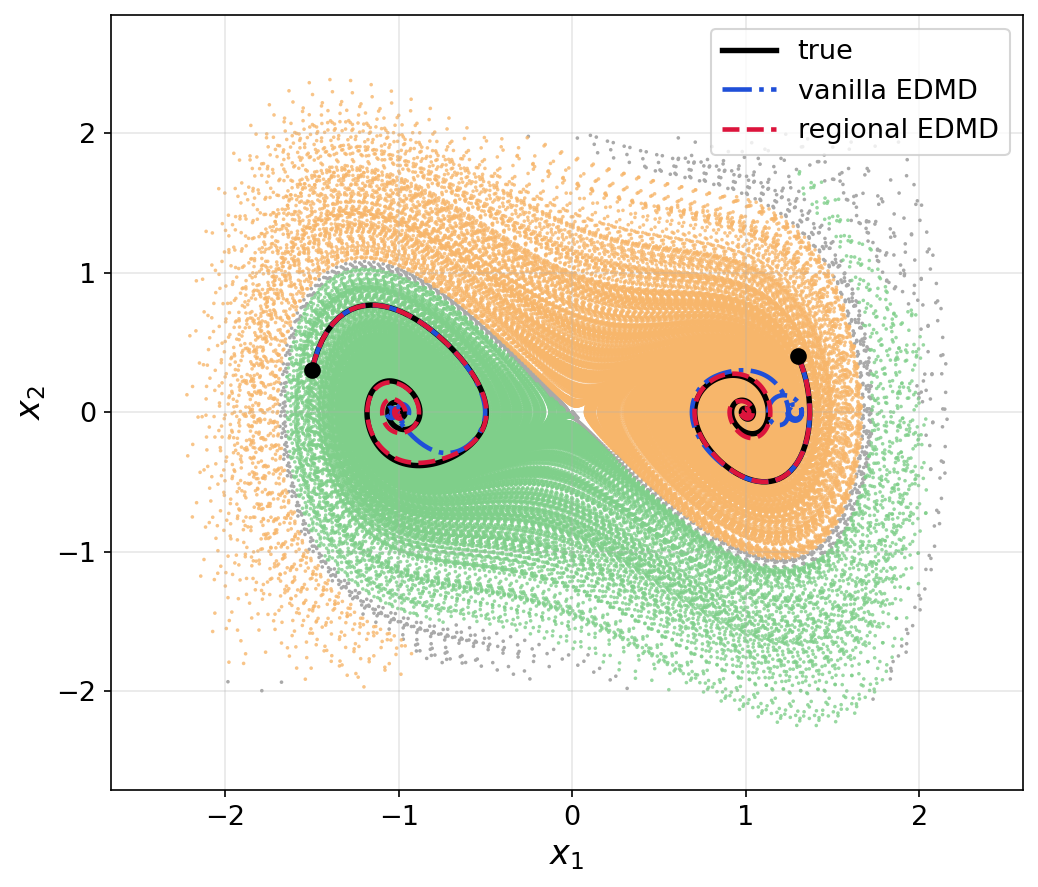}
    \caption{Duffing oscillator}
    \label{fig:duffing_switched}
  \end{subfigure}
  \hfill
  \begin{subfigure}[t]{0.49\linewidth}
    \centering
    \includegraphics[width=.9\linewidth]{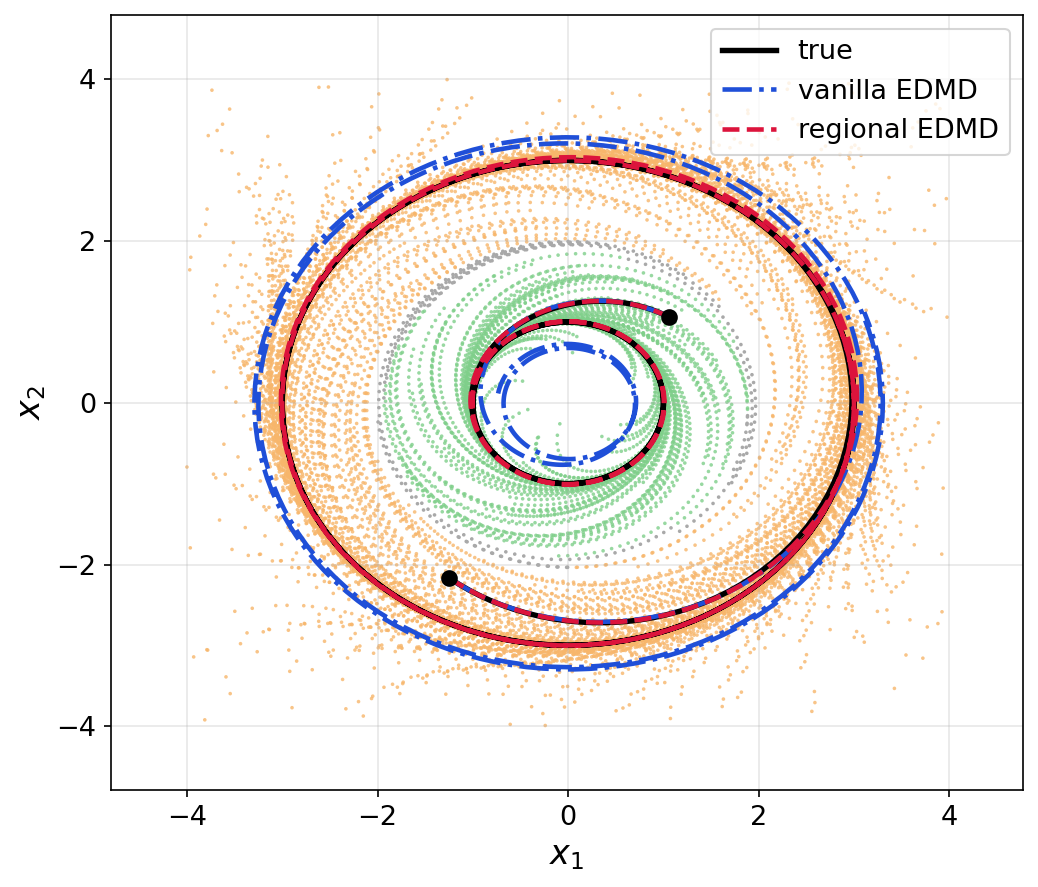}
    \caption{Bistable oscillations}
    \label{fig:bistable_switched}
  \end{subfigure}
  \caption{Recovered $\omega$-limit partition and multiple regional
    Koopman prediction for (a) the Duffing oscillator and (b) the
    bistable oscillations. Snapshots are colored (orange and green)
    based on the robustified cell $\widetilde D_\ell$; gray points are
    dropped by $\tau_{\mathrm{robust}}$. True (black), regional EDMD
    (red, dashed), and vanilla EDMD (blue, dash-dotted) trajectories
    start from two initial conditions.}
  \label{fig:phase_portraits}
\end{figure*}

\subsection{Setup}

Consider the damped Duffing oscillator
\begin{equation}
  \dot{x_1} = x_2,
  \,\,
  \dot{x_2} = -\delta x_2 + \alpha x_1 - \beta x_1^3,
\end{equation}
with parameters $\delta = 0.5$, $\alpha = 1.0$, and $\beta = 1.0$. The
system has two stable equilibria at $(\pm\sqrt{\alpha/\beta},0)$ and
an unstable equilibrium at $(0,0)$, so trajectories converge to
one of two equilibrium points.

Next, consider the bistable oscillations model
\begin{equation}
    \dot{r} = (r - 2) - (r - 2)^3,
    \,\,
    \dot{\theta} = 1.
\end{equation}
This system has two stable limit cycles at $r = 1, 3$ and an unstable
limit cycle at $r = 2$, so trajectories converge to one of two
disjoint $\omega$-limit sets. Initial conditions are drawn uniformly
in polar coordinates $(r,\theta)$; we convert to Cartesian coordinates
for data collection and partition learning, which makes the problem
harder.

In both examples, we simulate $500$ trajectories from random initial
conditions, sampling every $\Delta t=0.05$, to obtain one-step
snapshot pairs. To recover the partition
(Algorithm~\ref{alg:post_process}), we use a dictionary of $s=500$
Gaussian RBFs $\phi(x)=\exp(-\gamma\|x-c\|^2)$ with k-means centers
$c$, augmented by polynomials up to a prescribed degree. All
parameters are listed in Table~\ref{tab:parameters}.
Elements of the obtained partition $\widehat \Pc$ are visualized in
Figure~\ref{fig:phase_portraits} and Table~\ref{tab:validation}
reports the resulting classification accuracy on validation snapshots.

Each cell of the obtained partition $\widehat \Pc$ is then used to
build a regional EDMD model fit on the snapshots classified according
to the cell label, using a separate dictionary of RBF centers
subsampled from the training pool and polynomials. As a baseline, we
also fit a single (vanilla) EDMD model with the same dictionary on all
snapshot pairs, without using the
partition. Figure~\ref{fig:rollout_errors} shows the relative rollout
error on validation trajectories and Table~\ref{tab:validation}
reports the mean relative rollout error at selected time steps.

\begin{figure}[htb!]
  \centering
  \begin{subfigure}[t]{0.49\linewidth}
    \centering
    \includegraphics[width=\linewidth]{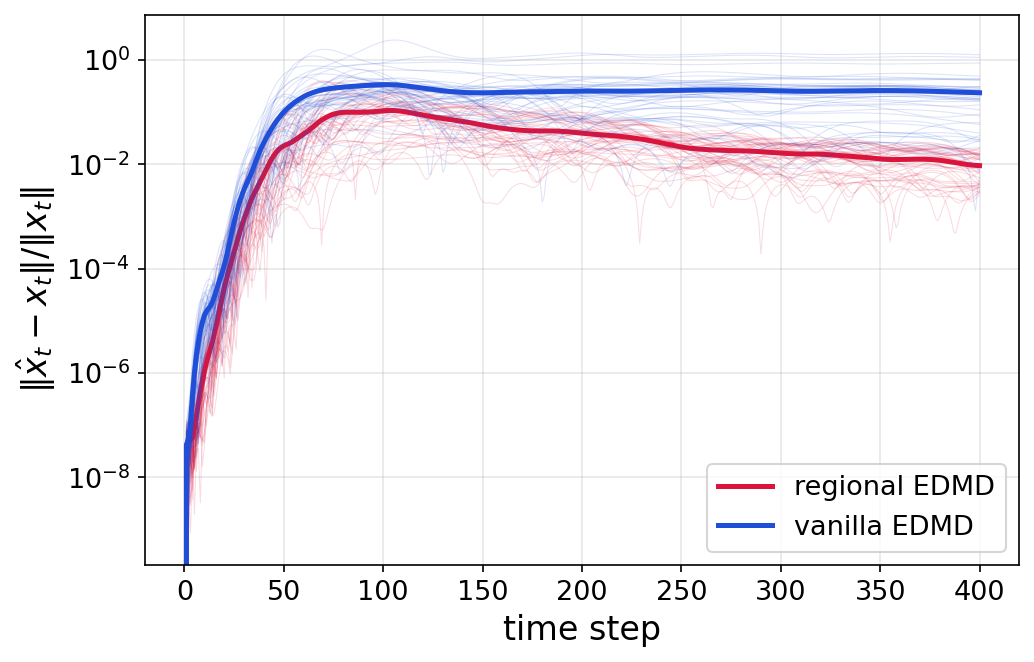}
    \caption{Duffing oscillator}
    \label{fig:duffing_errors}
  \end{subfigure}
  \hfill
  \begin{subfigure}[t]{0.49\linewidth}
    \centering
    \includegraphics[width=\linewidth]{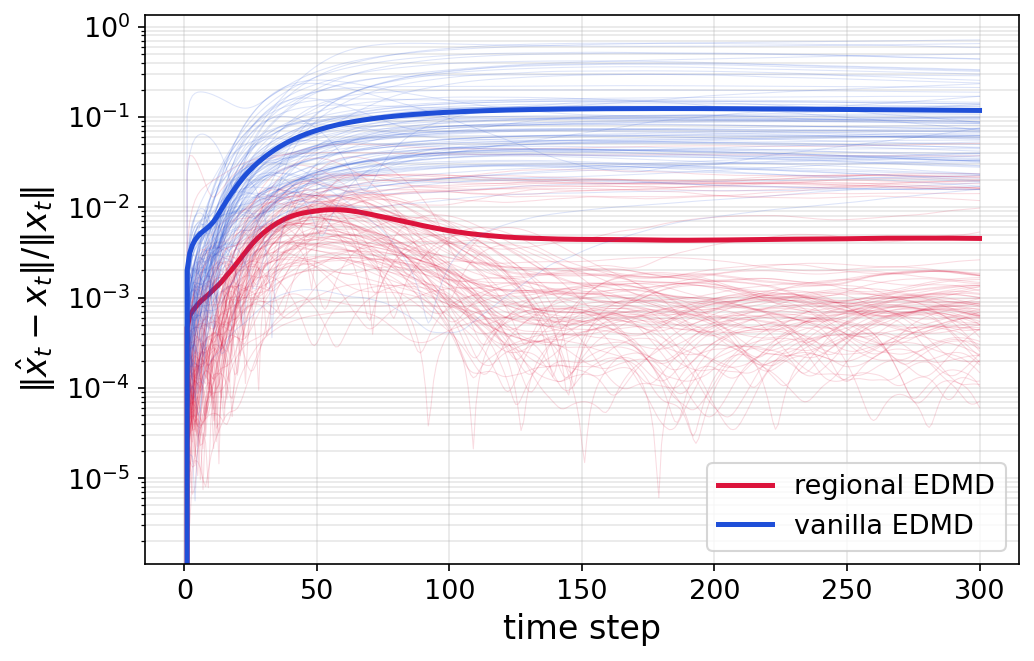}
    \caption{Bistable oscillations}
    \label{fig:bistable_errors}
  \end{subfigure}
  \caption{Relative multi-step rollout error for (a) the Duffing oscillator and
    (b) the bistable oscillations. Thin lines are
    individual trajectories and thick lines the mean over
    trajectories for the regional (red) and vanilla (blue) EDMD
    models.}
  \label{fig:rollout_errors}
\end{figure}
%
%

\begin{table}[htb!]
  \centering
  \small
  \begin{tabular}{@{}lllrr@{}}
    \toprule
    System & Model & Accuracy & $t\!=\!50$ & $t\!=\!300$ \\
    \midrule
    \multirow{3}{*}{Duffing} & regional & 100\% (0.42\%) & $22.6$ & $16.4$ \\
    & regional$^\ast$ & 99.86\% (0\%) & \multicolumn{2}{c}{diverged ($24$ of $50$)} \\
    & vanilla & --- & $100.0$ & $255.0$ \\
    \midrule
    \multirow{3}{*}{Bistable osc.} & regional & 100\% (0.23\%) & $9.2$ & $4.5$ \\
    & regional$^\ast$ & 99.90\% (0\%) & $18.8$ & $6.1$ \\
    & vanilla & --- & $71.3$ & $118.0$ \\
    \bottomrule
  \end{tabular}
  \caption{Classification accuracy of the recovered $\omega$-limit
    partition on the classified validation snapshots, with the
    fraction dropped by $\tau_{\mathrm{robust}}$ in parentheses. Other
    columns report mean relative state error ($\times 10^{-3}$)
    {averaged over several trajectories at time step $t$ for correctly
      classified initial conditions.} Rows marked $^\ast$ repeat the
    regional model with the robustness margin disabled
    ($\tau_{\mathrm{robust}}=0$).}
\label{tab:validation}
\end{table}

\begin{remark}[Parameter Tuning]
  We increase $\gamma$ until further increases no longer yield
  additional approximate invariant observables, i.e., until the number
  of eigenvectors
  %
  %
  of the GEP \eqref{eq:empirical_GEP} with eigenvalue at most
  $\epsilon^2$ stops growing.
  %
  %
  We then choose the scalar clustering integers $k_1,\ldots,k_r$. For
  the merging thresholds, we tune the transition threshold
  $\tau_{\mathrm{tr}}$ before using value-based merging. Specifically,
  we first set $\tau_{\mathrm{val}}=0$ so that merging is driven only
  by observed transitions between fine cells. We then vary
  $\tau_{\mathrm{tr}}$. Smaller values of $\tau_{\mathrm{tr}}$ produce
  more aggressive merging and, in the extreme, may collapse the
  partition to a single cell. Starting from a small value, we
  gradually increase $\tau_{\mathrm{tr}}$ until persistent cells
  appear and the resulting partition is approximately stationary on
  the data. After $\tau_{\mathrm{tr}}$ is fixed, the value-based
  threshold $\tau_{\mathrm{val}}$ can be used as a secondary
  refinement parameter.  \oprocend
\end{remark}

\subsection{Observations}

From Figure~\ref{fig:rollout_errors} and comparing the regional vs vanilla rows of Table~\ref{tab:validation}, we see an order of magnitude improvement in predictive performance when using regional EDMD models. In fact, the vanilla EDMD model has a large steady state error which can be attributed to trying to satisfy distinct asymptotic behaviors with one model. This is clearly visible in the phase portraits of Figure \ref{fig:phase_portraits} as well. 

Table \ref{tab:validation} demonstrates that correctly tuned $\tau_{\text{robust}}$ values lead to perfect accuracy for the confident samples. With
$\tau_{\mathrm{robust}}=0$, the partition still classifies more than
$99.8\%$ of the validation snapshots correctly, but the few
misclassified pairs near the boundary contaminate the regional EDMD models. For the Duffing oscillator, just $0.14\%$ of the fit data give one local operator an unstable mode,
and $24$ of the $50$ rollouts diverge even though their initial
conditions are correctly classified. This demonstrates that other naive approaches to building $\omega$-limit partitions might succeed in giving almost accurate partitions but in order to create useful regional EDMD models, one needs perfect partitions albeit outside an exception set of the original state space. 

\section{Conclusions}
Multiple omega-limit sets obstruct a single continuous
finite-dimensional Koopman model, but not a collection of them. OPAL
recovers the omega-limit partition from data, and our analysis explains
why this is possible: outside a small exception set, approximate
invariants err only by splitting cells, so merging fragments suffices.
Partition quality crucial: regional EDMD models improve
long-horizon prediction by an order of magnitude, yet a partition that
is $99.8\%$ correct can yield diverging models. Important open directions
include automatic dictionary and threshold selection for higher-dimensional
systems, performance guarantees for Algorithm \ref{alg:post_process}, and the integration of recovered
partitions into control design with Koopman representations.

\bibliographystyle{ieeetr}
\bibliography{../bib/alias,../bib/Main-add,../bib/Main,../bib/JC}

\newpage

\end{document}